\def\year{2018}\relax
\documentclass[letterpaper]{article} 
\usepackage{aaai18}  
\usepackage{times}  
\usepackage{helvet}  
\usepackage{courier}  
\usepackage{url}  
\usepackage{graphicx}  
\newcommand{\citep}[1]{\citeauthor{#1} [\citeyear{#1}]}
\usepackage{amsmath}
\usepackage{amsfonts}
\usepackage{amssymb}
\usepackage{amsthm}
\newtheorem{theorem}{Theorem}
\usepackage{epstopdf}
\usepackage{multicol}
\usepackage{multirow}
\usepackage{float}

\usepackage[dvipsnames]{xcolor}

\usepackage{mathtools}
\DeclarePairedDelimiter\ceil{\lceil}{\rceil}
\DeclarePairedDelimiter\floor{\lfloor}{\rfloor}

\newcommand*{\affaddr}[1]{#1} 
\newcommand*{\affmark}[1][*]{\textsuperscript{#1}}
\newcommand*{\email}[1]{\texttt{#1}}

\begin{document}
%
\title{Sequence-to-point learning with neural networks for non-intrusive load monitoring}
\author{AAAI Press\\
Association for the Advancement of Artificial Intelligence\\
2275 East Bayshore Road, Suite 160\\
Palo Alto, California 94303\\
}
\author{
Chaoyun Zhang\affmark[1], Mingjun Zhong\affmark[2], Zongzuo Wang\affmark[1], Nigel Goddard\affmark[1], and Charles Sutton\affmark[1]\\
\affaddr{
\affaddr{\affmark[1]School of Informatics, University of Edinburgh, United Kingdom}\\
\email{chaoyun.zhang@ed.ac.uk}, \email{\{ngoddard,csutton\}@inf.ed.ac.uk}\\
\affmark[2]School of Computer Science, University of Lincoln, United Kingdom}\\
\email{mzhong@lincoln.ac.uk}\\
}
\maketitle
\begin{abstract}
Energy disaggregation (a.k.a nonintrusive load monitoring, NILM), a
single-channel blind source separation problem, aims to decompose the
mains which records the whole house electricity consumption into
appliance-wise readings. This problem is difficult because it is
inherently unidentifiable. Recent approaches have shown that the
identifiability problem could be reduced by introducing domain
knowledge into the model. Deep neural networks have been shown to be a
promising approach for these problems, but sliding windows are
necessary to handle the long sequences which arise in signal
processing problems, which raises issues about how to combine
predictions from different sliding windows.  In this paper, we propose
sequence-to-point learning, where the input is a window of the mains
and the output is a single point of the target appliance. We use
convolutional neural networks to train the model. Interestingly, we
systematically show that the convolutional neural networks can
inherently learn the signatures of the target appliances, which are
automatically added into the model to reduce the identifiability
problem. We applied the proposed neural network approaches to
real-world household energy data, and show that the methods achieve
state-of-the-art performance, improving two standard error measures by
84\% and 92\%.
\end{abstract}

\noindent Energy disaggregation \cite{hart92} is a single-channel blind source separation (BSS) problem
that aims to decompose the whole energy consumption of a dwelling
into the energy usage of individual appliances. 
The purpose is to help households to reduce their energy consumption
by helping them to understand what is causing them to use energy, and 
it has been shown that  disaggregated information can help  householders to reduce energy consumption by 
as much as $5-15\%$ \cite{fischer2008feedback}.
However, current electricity meters can only report the whole-home consumption data. This triggers the demand of machine-learning tools to infer the appliance-specific consumption.

Energy disaggregation is unidentifiable and thus a difficult prediction problem because it is a single-channel BSS problem; we want to extract more than one source from a single observation. 
Additionally, there are a large number of sources of uncertainty in the prediction problem,
including noise in the data, lack of knowledge of the true power usage for every appliance in a given household,
 multiple devices exhibiting similar power consumption, and simultaneous switching on/off of multiple devices.
Therefore energy disaggregation has been an active area for the application of artificial intelligence
and machine learning techniques.
Popular approaches have been based on factorial hidden Markov models (FHMM) \cite{kolter2012approximate,Parson12,zhong13,zhong2014signal,zhong2015latent,lange2016efficient} and signal processing methods \cite{pattem2012unsupervised,Zhao15,Zhao16,Batra16,Tabatabaei17}.

Recently, it has been shown that single-channel BSS could be modelled by using sequence-to-sequence (seq2seq) learning with neural networks \cite{grais2014deep,huang2014deep,Du16}. In particular, it has been applied to energy disaggregation \cite{kelly2015neural} \textemdash both convolutional (CNN) and recurrent neural networks (RNN) were employed. The idea of sequence-to-sequence learning is to train a deep network to map between an input sequence, such as the mains power readings in the NILM problem,
and an output sequence, such as the power readings of a single appliance.

A difficulty immediately arises when applying seq2seq in signal
processing applications such as BSS.  In these applications, the input
and output sequences can be long, for example, in one of our data
sets, the input and output sequences are 14,400 time steps.  Such long
sequences can make training both computationally difficult, both
because of memory limitations in current graphics processing units
(GPUs) and, with RNNs, because of the vanishing gradient problem.  A
common way to avoid these problems is a sliding window approach, that
is, training the network to map a window of the input signal to the
corresponding window of the output signal.  However, this approach has
several difficulties, in that each element of the output signal is
predicted many times, once for each sliding window; an average of
multiple predictions is naturally used, which consequently smooths the
edges.  Further, we expect that some of the sliding windows will
provide a better prediction of a single element than others;
particularly, those windows where the element is near the midpoint of
the window rather than the edges, so that the network can make use of
all nearby regions of the input signal, past and future.  But a simple
sliding window approach cannot exploit this information.


In this paper, we propose a different architecture
called \emph{sequence-to-point learning (seq2point)} for
single-channel BSS. This uses a sliding window approach, but given a
window of the input sequence, the network is trained to predict the
output signal only at the midpoint of the window.  This has the effect
of making the prediction problem easier on the network, as rather than
needing to predict in total $W(T-W)$ outputs as in the seq2seq method,
where $T$ is the length of the input signal and $W$ the size of the
sliding window, the seq2point method predicts only $T$ outputs. This
allows the neural network to focus its representational power on the
midpoint of the window, rather than on the more difficult outputs on
the edges, hopefully yielding more accurate predictions.


We provide both an analytical and empirical analysis of the methods,
showing that seq2point has a tighter approximation to the target
distribution than seq2seq learning.  On two different real-world NILM
data sets (UK-DALE \cite{kelly2014uk} and REDD \cite{kolter2011redd}),
we find that sequence-to-point learning performs dramatically better
than previous work, with as much as 83\% reduction in error.

Finally, to have confidence in the models, it is vital to interpret
the model predictions and understand what information the neural
networks for NILM are relying on to make their predictions.  By
visualizing the feature maps learned by our networks, we found that
our networks automatically extract useful features of the input
signal, such as change points, and typical usage durations and power
levels of appliances.  Interestingly, these signatures have been
commonly incorporated into handcrafted features and architectures for
the NILM problem
\cite{kolter2012approximate,Parson12,pattem2012unsupervised,Zhao15,zhong2014signal,zhong2015latent,Batra16,Tabatabaei17},
but in our work these features are learned automatically.

\section{Energy disaggregation}
The goal of energy disaggregation is to recover the energy consumption
of individual appliances from the mains signal, which measures the
total electricity consumption.  Suppose we have observed the mains $Y$
which indicates the total power in Watts in a household, where $Y =
(y_1, y_2, ... ,y_T)$ and $y_t \in R_{+}$. Suppose there are a number
of appliances in the same house. For each appliance, its reading is
denoted by $X_i = (x_{i1}, x_{i2}, ..., x_{iT})$, where $x_{it} \in
R_{+}$. At each time step, $y_t$ is assumed to be the sum of the
readings of individual appliances, possibly plus a Gaussian noise
factor with zero mean and variance $\sigma^2$ such that $ y_t
= \sum_{i} x_{it} + \epsilon_t $.  Often we are only interested in $I$
appliances, i.e., the ones that use the most energy; others will be
regarded as an unknown factor $U =(u_1,\cdots,u_T)$. The model
could then be represented as $y_t = \sum_{i=1}^I x_{it} + u_t
+ \epsilon_t$.


The additive factorial hidden Markov model (AFHMM) is a natural approach to represent this model  \cite{kolter2012approximate,pattem2012unsupervised,zhong2014signal}. Various inference algorithms could then be employed to infer 
the appliance signals $\{ X_i \}$ \cite{kolter2012approximate,zhong2014signal,Kiarash16}. It is well-known that the problem is still unidentifiable. 
 To tackle the identifiability problem, various approaches have been proposed by incorporating domain knowledge into the model. For example, local information, e.g., appliance power levels, ON-OFF state changes, and durations, has been incorporated into the model \cite{kolter2012approximate,Parson12,pattem2012unsupervised,Zhao15,Tabatabaei17}; others have incorporated global information, e.g., total number of cycles and total energy consumption \cite{zhong2014signal,zhong2015latent,Batra16}. However, the domain knowledge required by these methods needs to be extracted manually, which makes the methods more difficult to use.
 As was previously noted, all these approaches require handcrafted features based on the observation data. Instead, we propose to employ neural networks to extract those features automatically during learning.

\section{Sequence-to-sequence learning}

\citep{kelly2015neural} have applied deep learning methods to NILM. The neural networks learns a nonlinear regression between a sequence of the mains
readings and 
a sequence of appliance readings {\em with the same time stamps}. 
We will refer to this as a \emph{sequence-to-sequence} approach. Although RNN architectures
are most commonly used in sequence-to-sequence
learning for text \cite{sutskever14}, 
for NILM \citep{kelly2015neural} employ
both CNNs and RNNs. 
Similar sequence-to-sequence neural network approaches have been applied
 to single-channel BSS problems in audio and
 speech \cite{grais2014deep,huang2014deep,Du16}.
 
Sequence-to-sequence architectures define a neural
network $F_s$ that maps sliding windows $Y_{t:t+W-1}$ of the input
mains power to corresponding windows $X_{t:t+W-1}$ of the output
appliance power, that is, they model $X_{t:t+W-1} = F_s(Y_{t:t+W-1})
+ \epsilon$, where $\epsilon$ is $W$-dimensional Gaussian random
noise.
Then, to train the network on a pair $(X,Y)$ of 
full sequences, the loss function is
\begin{equation}
\label{seq2seq}
	L_s = \sum_{t=1}^{T-W+1} \log p(X_{t:t+W-1} | Y_{t:t+W-1}, \theta_s),
\end{equation}
where $\theta_s$ are the parameters of the network $F_s$.
In practice, a subset of all possible windows can be used during
 training in order to reduce computational complexity.
 
Since there are multiple predictions for ${x}_t$ when $2\leq
 t\leq{T}-1$, one for each sliding window that contains time $t$,
 the mean of these predicted values is used as the
 prediction result.  It has been shown that this neural network
 approach outperforms AFHMMs for the NILM task. 

\section{Sequence-to-point learning}

\begin{figure*}[htbp!]
\centering
\includegraphics[width=0.9\textwidth]{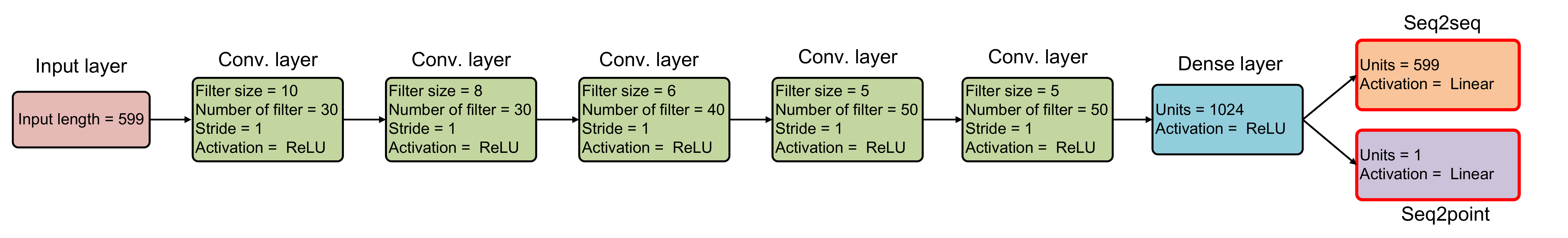}%
\vspace{-1.5em}\caption{\label{Fig:cnn} The architectures for sequence-to-point and sequence-to-sequence neural networks.}
\end{figure*}

Instead of training a network to predict a window of appliance readings, we propose to train a neural network to only predict the midpoint element of that window. The idea is that the input of the network is a mains window
$Y_{t:t+W-1}$, and the output is the midpoint element $x_\tau$ of the corresponding window of the target appliance,
where $\tau = t + \floor{W/2}$.
 We call this type of method a sequence-to-point learning method, which is widely applied for modelling the distributions of speech and image \cite{sainath15,oord2016wavenet,oord2016pixel}. This method assumes that the midpoint element is represented as a non-linear regression of the mains window. The intuition behind this assumption is that we expect the state of the midpoint element of that appliance should relate to the information of mains before and after that midpoint. We will show explicitly in the experiments that the change points (or edges) in the mains are among
the features that the network uses to infer the states of the appliance.


Instead of mapping sequence to sequence, the seq2point architectures
define a neural network $F_p$ which maps sliding windows $Y_{t:t+W-1}$ of
the input to the midpoint $x_{\tau}$ of the corresponding windows
$X_{t:t+W-1}$ of the output. The model is
$x_{\tau}=F_p(Y_{t:t+W-1})+\epsilon$. The loss function has the following
form for training the network:
\begin{equation}
\label{seq2point}
L_p = \sum_{t=1}^{T-W+1} \log p(x_{\tau} | Y_{t:t+W-1}, \theta_p).
\end{equation}
where $\theta_p$ are the network parameters. 
 To deal with
the endpoints of the sequence, given a full 
input sequence $Y = (y_1 \ldots y_T)$, we first pad the sequence
with $\ceil{W/2}$ zeros at the beginning and end. 
The advantage of the seq2point model is that 
there is a single prediction for every $x_t$, rather than
an average of predictions for each window.

\subsection{Architectures}
\citep{kelly2015neural} showed that denoising autoencoders performed better than other architectures for seq2seq learning. Instead, we propose to employ the same convolutional neural network for seq2seq and seq2point learning in this paper, as we will show in the following section that both approaches can have the same architecture, shown in Figure \ref{Fig:cnn}. \citep{kelly2015neural} generated the training data heuristically. In contrast, we use all the sliding windows for both methods for training, thus not requiring heuristic rules for generating training data.

\begin{table*}[tp]
\centering
\caption{\label{on} The parameters used in this paper for each appliance. Power unit is Watt.}
\label{my-label}
\begin{tabular}{c|c|c|c|c|c }
& Kettle & Microwave & Fridge & Dish Washer & Washing Machine \\ \hline
Window length (point) & 599 &599 & 599 & 599 & 599  \\
Maximum power & 3948 & 3138 & 2572 & 3230 & 3962 \\
On power threshold & 2000 & 200 & 50 & 10& 20 \\
Mean on power  & 700 & 500 & 200 & 700 & 400 \\
Standard deviation on power & 1000 & 800 & 400 & 1000 & 700
\end{tabular}
\end{table*}

\subsection{Posterior distribution estimators}
In this section we show that both seq2seq and seq2point learning are essentially posterior density estimators. Suppose $T\rightarrow\infty$, then we could have infinite number of sliding windows. They inherently form a population distribution $\pi(X|Y)$ which is unobserved, where $X$ and $Y$ are random temporally-aligned vectors of length $W$. Seq2seq tries to find $\theta$ to maximise $p(X|Y,\theta)$ to approximate the population posterior $\pi(X|Y)$. This could be achieved by minimizing Kullback-Leibler (KL) divergence with respect to $\theta$
\begin{equation}
\min_{\theta}KL(\pi||p)=\min_{\theta}\int \pi(X|Y)\log{\frac{\pi(X|Y)}{p(X|Y,\theta)}}dX.\nonumber
\end{equation}
So we have the standard interpretation that both methods are minimizing
a Monte Carlo approximation to KL-divergence.  

Now we will characterize the difference between seq2seq and seq2point learning. 
If we assume a factorizable form such that $p(X|Y,\theta)=\prod_{w=1}^Wp_w(x_w|Y,\theta)$, the objective function can then be represented as
\begin{equation}
KL(\pi||p)=\sum_{w=1}^WKL(\pi(x_w|Y)||p_w(x_w|Y,\theta)).\nonumber
\end{equation}
Now denote $\phi_w(\theta|Y)=KL(\pi(x_w|Y)||p_w(x_w|Y,\theta))$ which is a function of $\theta$ given $Y$.

Seq2seq learning assumes the following distribution
\begin{align}
p(X|Y,\theta) & =\mathcal{N}(\mu(\theta),cI) =\prod_{w=1}^W\mathcal{N}(\mu_w(\theta),c)\nonumber \\
&=\prod_{w=1}^Wp_w(x_w|Y,\theta)\nonumber
\end{align}
where $\mu(\theta)=(\mu_1(\theta),\cdots,\mu_W(\theta))^T$, $c$ is a constant, and $I$ is the identity matrix. All the distributions $p_w$ partially share the same parameters $\theta$ except the parameters from the last hidden layer to outputs, and therefore, optimization needs to be performed jointly over all the distributions $p_w$ ($w=1,2,\cdots,W$) such that
\begin{equation}
\min_{\theta}\sum_{w=1}^W\phi_w(\theta|Y).\nonumber
\end{equation}
Denote $\widetilde{\theta}$ as the optimum, the approximate distribution of the midpoint value is then $p_{\tau}(x_{\tau}|Y,\widetilde{\theta})$ by using seq2seq, and the corresponding KL-divergence for the midpoint value is $\phi_{\tau}(\widetilde{\theta}|Y)=KL(\pi(x_{\tau}|Y)||p_{\tau}(x_{\tau}|Y,\widetilde{\theta}))$.

 The seq2point learning directly models the midpoint value, and therefore, the optimization over $p(x_{\tau}|Y,\theta)$ can be performed by the following problem
\begin{equation}
\min_{\theta}\phi_{\tau}(\theta|Y).\nonumber
\end{equation}
Denote $\theta^*$ as the optimum, the approximate distribution of the midpoint value is then $p_{\tau}(x_{\tau}|Y,{\theta}^*)$. 
This shows that seq2seq and seq2point infer two different approximate distributions respectively to the same posterior distribution for midpoint value. 

The following theorem shows that seq2point learning infers a tighter approximation to the target distribution than the seq2seq learning when they use the same architecture.
\begin{theorem}
Assume all the distributions are well-defined. Suppose both the seq2point and seq2seq learning have the same architecture. Suppose $\theta^*$ is the optimum of the seq2point model, and $\widetilde{\theta}$ is the optimum of the seq2seq model. Then $\phi_{\tau}(\theta^*|Y)\leq \phi_{\tau}(\widetilde{\theta}|Y)$.
\end{theorem}
\begin{proof}
It is natural to assume that all the distributions are well-defined, and thus KL-divergence has a lower bound $0$ such that $\phi_w\geq 0$. Since both learning methods have the same architecture, the functions $\phi_w$ for the two methods are the same. Since $\theta^*$ is the optimum of the problem $\min_{\theta}\phi_{\tau}(\theta|Y)$, for any $\theta$, $\phi_{\tau}(\theta^*|Y)\leq \phi_{\tau}(\theta|Y)$. So $\phi_{\tau}(\theta^*|Y) + \sum_{w=1,w\neq \tau}^W\phi_w(\theta|Y)\leq \sum_{w=1}^W\phi_w(\theta|Y)$ is true for any $\theta$. Therefore, $\phi_{\tau}(\theta^*|Y) + \sum_{w=1,w\neq \tau}^W\phi_w(\widetilde{\theta}|Y)\leq \sum_{w=1}^W\phi_w(\widetilde{\theta}|Y)$. Consequently, $\phi_{\tau}(\theta^*|Y)\leq \phi_{\tau}(\widetilde{\theta}|Y)$.
\end{proof}
This theorem ensures that seq2point learning always provides a tighter approximation than seq2seq learning. 

\begin{table*}[tp]
\centering
\caption{The appliance-level mean absolute error (MAE) (Watt) and signal aggregate error (SAE) for UK-DALE data. Best results are shown in bold. Seq2seq(Kelly) is proposed in \cite{kelly2015neural}.}
\label{ukdale}
\begin{tabular}{c|c|cccccc}
Error measures&\textbf{Methods}   & Kettle & Microwave& Fridge & Dish w. & Washing m. & Overall     \\
\hline
\multirow{4}{*}{MAE}&AFHMM & 47.38& 21.18& 42.35 &199.84 &103.24 &82.79 $\pm$ 64.50\\
&seq2seq(Kelly)                       & 13.000 & 14.559 & 38.451& 237.96& 163.468& 93.488  $\pm$ 91.112    \\
&seq2seq(this paper)                          & 9.220 & 13.619  & 24.489 & 32.515 & \textbf{10.153} & 17.999 $\pm$ 9.063  \\
&seq2point(this paper)                          & \textbf{7.439} & \textbf{8.661} & \textbf{20.894} & \textbf{27.704} & 12.663& $\mathbf{15.472  \pm 7.718}$    \\

\hline
\multirow{4}{*}{SAE}&AFHMM & 1.06 &1.04 &0.98& 4.50& 8.28& 3.17 $\pm$ 2.88\\
&seq2seq(Kelly)                       & 0.085 & 1.348 & 0.502& 4.237& 13.831& 4.001  $\pm$ 5.124    \\
&seq2seq(this paper)                           & 0.309 & \textbf{0.205}  & 0.373 & 0.779 & 0.453 & 0.423 $\pm$ 0.194  \\
&seq2point(this paper)                         & \textbf{0.069} & 0.486 & \textbf{0.121} & \textbf{0.645} & \textbf{0.284} & $\mathbf{0.321  \pm 0.217}$    \\

\end{tabular}
\end{table*}

\begin{table*}[tp]
\centering
\caption{The appliance-level mean absolute error (MAE) (Watt) and signal aggregate error (SAE) for REDD data. Best results are shown in bold.}
\label{redd}
\begin{tabular}{c|c|cccccc}
Error measures&\textbf{Methods}   & Microwave& Fridge & Dish w. & Washing m. & Overall     \\
\hline
\multirow{2}{*}{MAE}&seq2seq(this paper)                          &  33.272  & 30.630 & \textbf{19.449} & {22.857} & 26.552 $\pm$ 5.610  \\
&seq2point(this paper)                          & \textbf{28.199} & \textbf{28.104} & {20.048} & \textbf{18.423}& $\mathbf{23.693  \pm 4.494}$    \\
\hline
\multirow{2}{*}{SAE}&seq2seq(this paper)                           & {0.242}  & \textbf{0.114} & \textbf{0.557} & 0.509 & 0.355 $\pm$ 0.183  \\
&seq2point(this paper)                          & \textbf{0.059} & {0.180} & \textbf{0.567} & \textbf{0.277} & $\mathbf{0.270  \pm 0.187}$    \\
\end{tabular}
\end{table*}

\begin{figure*}[t!]
\centering
\includegraphics[width=0.9\textwidth]{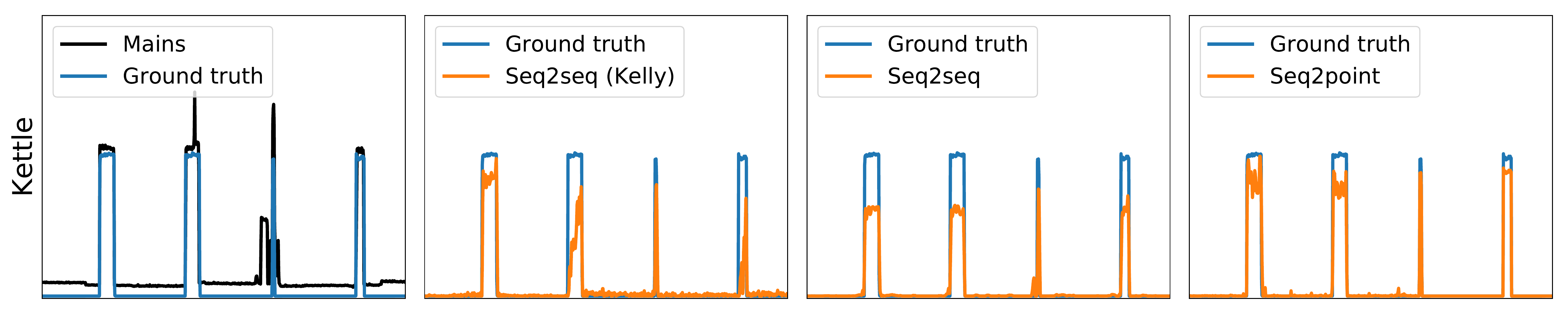}
\includegraphics[width=0.9\textwidth]{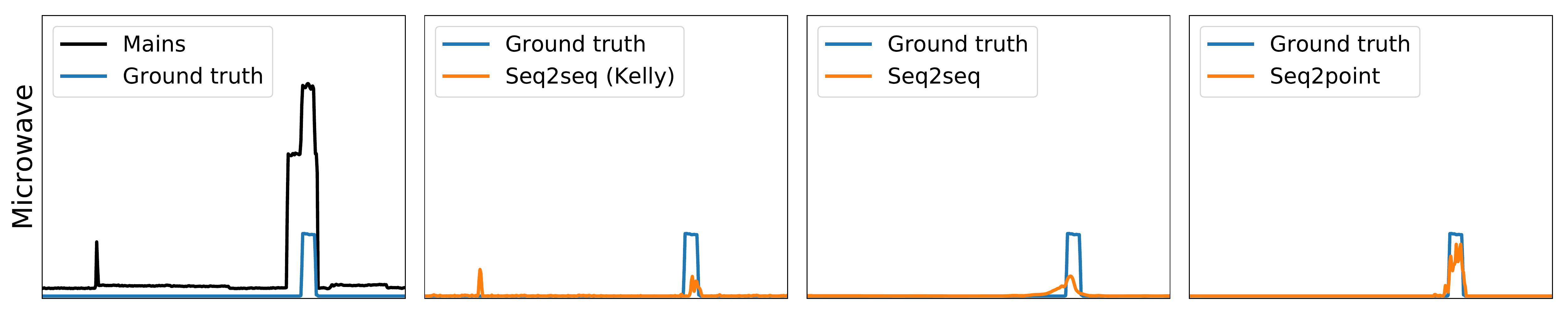}
\includegraphics[width=0.9\textwidth]{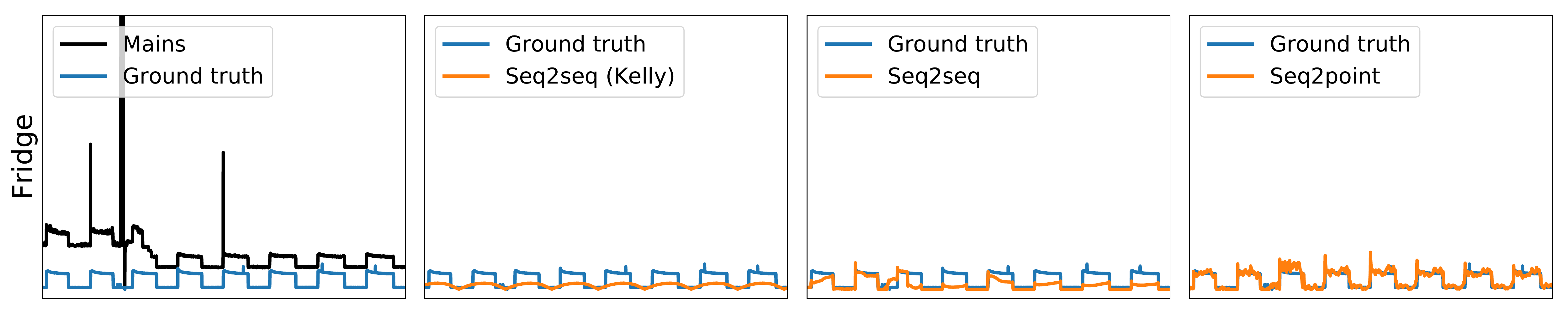}
\includegraphics[width=0.9\textwidth]{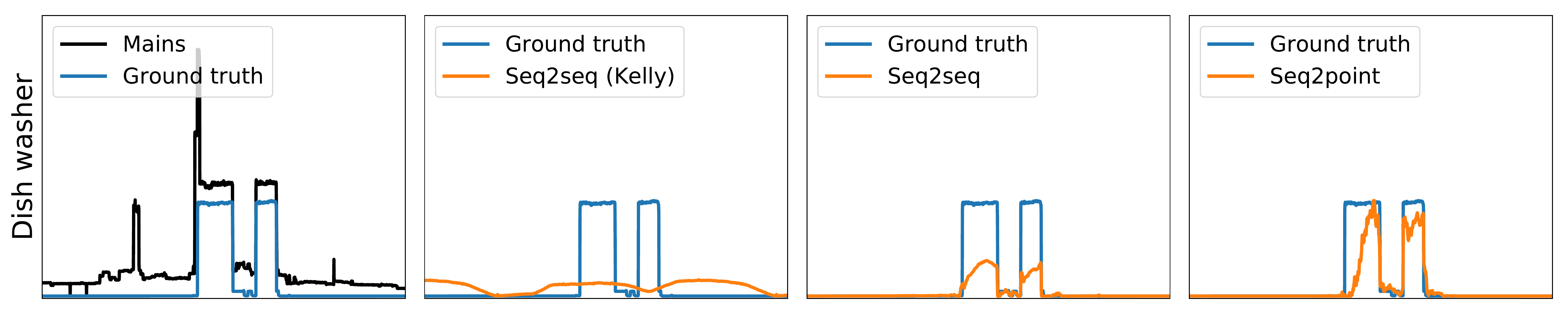}
\includegraphics[width=0.9\textwidth]{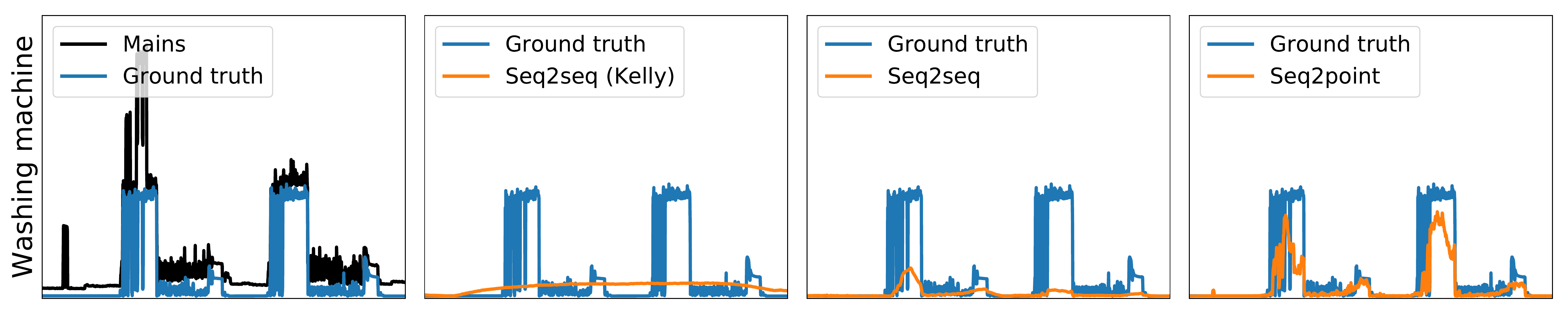}
\caption{\label{results} Some example disaggregation results on UK-DALE. Both seq2seq and seq2point are the methods proposed in this paper. Seq2seq(Kelly) is proposed in \cite{kelly2015neural}.} 
\end{figure*}

\section{Experiments}
\label{experiment}

We compare four different models for the energy disaggregation problem, namely, the AFHMM 
\cite{kolter2012approximate}, seq2seq(Kelly) \cite{kelly2015neural}, seq2seq, and seq2point. Note that seq2seq and seq2point use the same architecture (see Figure \ref{Fig:cnn}). There are two differences between the seq2seq proposed in this paper and the seq2seq(Kelly): 1) seq2seq uses the same training samples as seq2point where the samples were obtained by sliding the windows across all the data sequences; seq2seq(Kelly) uses selected windows obtained from all the data sequences, 
including some generated via data augmentation; 2) seq2seq uses a multilayer CNN architecture; seq2seq(Kelly) uses an autoencoder which includes a convolutional layer at each end. To verify the effectiveness and efficiency, we conduct comprehensive comparisons in terms of different performance metrics. The deep learning models are implemented in Python using TensorFlow. The networks were trained on machines with NVIDIA GTX 970 and NVIDIA GTX TITAN X GPUs. 

\subsection{Data sets}
We report results on the UK-DALE \cite{kelly2014uk} and REDD \cite{kolter2011redd} data sets, which measured the domestic appliance-level energy consumption and whole-house energy usage of five UK houses and six US houses respectively. 
\subsubsection{UK-DALE data}
All the readings were recorded in every 6 seconds from November 2012 to January 2015. The dataset contains the measurements of over 10 types of appliances, however, in this paper we are only interested in kettle, microwave, fridge, dish washer and washing machine which are popular appliances for evaluating NILM algorithms. We used the houses 1, 3, 4, and 5 for training the neural networks, and house 2 as the test data, because only houses 1 and 2 have all these appliances \cite{kelly2015neural,zhong2015latent}. 
Note that we are therefore considering the transfer learning setting in which we train and test on different households. This setting has the challenge that the same type of appliance will vary in its power 
demands in different houses, but good performance in the transfer learning set-up is vital
to practical application of NILM methods.
\subsubsection{REDD data}
The appliance and mains readings were recorded in every 3 seconds and 1 second respectively. The data set contains
measurements from six houses. We used houses 2 to 6 for training, and house 1 for testing the algorithms, for similar reasons to those in \citep{kelly2015neural}.
 Since there is no kettle data, we only looked at microwave, fridge, dish washer and washing machine.

\subsection{Data preprocessing}
We describe how the training data were prepared for training the neural networks. A window of the mains was used as the input sequence; the window length for each appliance is shown in the Table \ref{on}. The training windows were obtained by sliding the mains (input) and appliance (output) readings one timestep at a time; for seq2point, the midpoint values of the corresponding appliance windows were used as the outputs. Both the input windows and targets were preprocessed by subtracting the mean values and dividing by the standard deviations (see these parameters in the Table \ref{on}). These data samples were used for training both the seq2seq and seq2point methods. The training samples for training seq2seq(Kelly) were obtained by the method described in \citep{kelly2015neural}.

\subsection{Performance evaluation}
We use two metrics to compare these approaches. Denote $x_t$  as the ground truth and  $\hat{x}_t$ the prediction of an appliance at time $t$. 
When we are interested in the error in power at every time point, we use the mean absolute error (MAE)
\begin{equation}
\begin{aligned}
\mbox{MAE} = \frac{1}{T}\sum_{t=1}^T |\hat{x}_t-x_t|.\nonumber
\end{aligned}
\end{equation}
This provides a measure of errors that is less affected by outliers, i.e. isolated predictions that are particularly
inaccurate.
When we are interested in the total error in energy over a period, in this case, one day, we use the normalised signal aggregate error (SAE)
\begin{equation}
\begin{aligned}
\mbox{SAE} = \frac{|\hat{r}-r|}{r}\end{aligned},\nonumber
\end{equation}
where $r$ and $\hat{r}$ denote the ground truth and inferred total energy consumption of an appliance,
that is $r = \sum_t x_t$ and $\hat{r} = \sum_t \hat{x}_t.$ This measure is useful because a method
could be accurate enough for reports of daily power usage even if its per-timestep prediction
is less accurate.

\subsection{Experimental results}
First, on the UK-DALE data, Table \ref{ukdale} shows that both the
seq2seq and seq2point methods proposed in the paper outperformed the
other two methods (AFHMM and seq2seq(Kelly).  Our seq2seq reduces MAE
by 81\% and SAE by 89\% overall compared to seq2seq(Kelly), with
improvements in MAE for every appliance --- this can be explained by
our use of deeper architectures.  Our seq2point method
outperformed our seq2seq method in three out of four appliances, and
overall --- matching the results we obtained in the theorem. Compared
to seq2seq(Kelly) our seq2point reduces MAE by 84\% and SAE by 92\%.
We show example disaggregations on this data set performed by the
three neural network methods in Figure \ref{results}.

Since AFHMM and seq2seq(Kelly) perform worse than our two methods on
UK-DALE, we only applied our seq2seq and seq2point method to the REDD
data set. The results are shown in Table \ref{redd}. We can see
that seq2point outperformed seq2seq in most of the appliances, and
overall seq2point performs better than seq2seq --- improving MAE by
11\% and SAE by 24\%, very close to the overall improvements on
UK-DALE.

%

\subsection{Visualization of latent features}
\label{feature}

To validate the models, we would like to understand the reasons behind the network's predictions.
We expect that appliance signals have characteristic signatures that indicate when they are on.
For example, a kettle only has two states: ON and OFF, and when it is ON the power should be approximately $2,000-3,000$ Watts; 2) the approximate duration of the kettle when it is ON. This information could be enough to detect a kettle.
This information can greatly improve the performance of some algorithms \cite{Zhao15,Batra16,zhong2015latent}. 

Interestingly, we observed that the convolutional neural networks
proposed in this paper are inherently learning these signatures.
To test what the network has learnt, 
we take a window from the data, and manually modify it in ways
that we believe should affect the prediction of the network. In these
experiments we looked at the kettle which is easier to study because
there are less number of states. For each different input, we plotted
the feature maps of the last CNN layer in the
Figure \ref{Fig:analysis}. It is interesting that all the filters
detected the state changes of the appliances. More specifically, some
filters take the responsibility of detecting amplitude of the
appliance and as well as the state changes, but others only detect the
state changes. Figure \ref{Fig:analysis} (b) shows that when the
kettle was manually removed, the network suggests that the amplitude
of the signal and as well as the duration were not appropriate for a
kettle. Figure \ref{Fig:analysis} (c) shows that when the amplitude of
the kettle was double, the network detects the kettle which is
reasonable because both the duration and amplitude correspond to a
kettle. Figure \ref{Fig:analysis} (d) indicates that when the
amplitude of the kettle was manually reduced, the network suggests
there was no kettle. Figure \ref{Fig:analysis} (e) shows that when the
duration of the appliance usage was set too long ($>8$ minutes), the
network might suggest it was too long for a
kettle. Figure \ref{Fig:analysis} (f) shows that when there is no activation at the midpoint, the learnt signatures have the similar types to those in (a). 

\begin{figure*}[!htb]
\begin{center}
\includegraphics[width=0.8\textwidth]{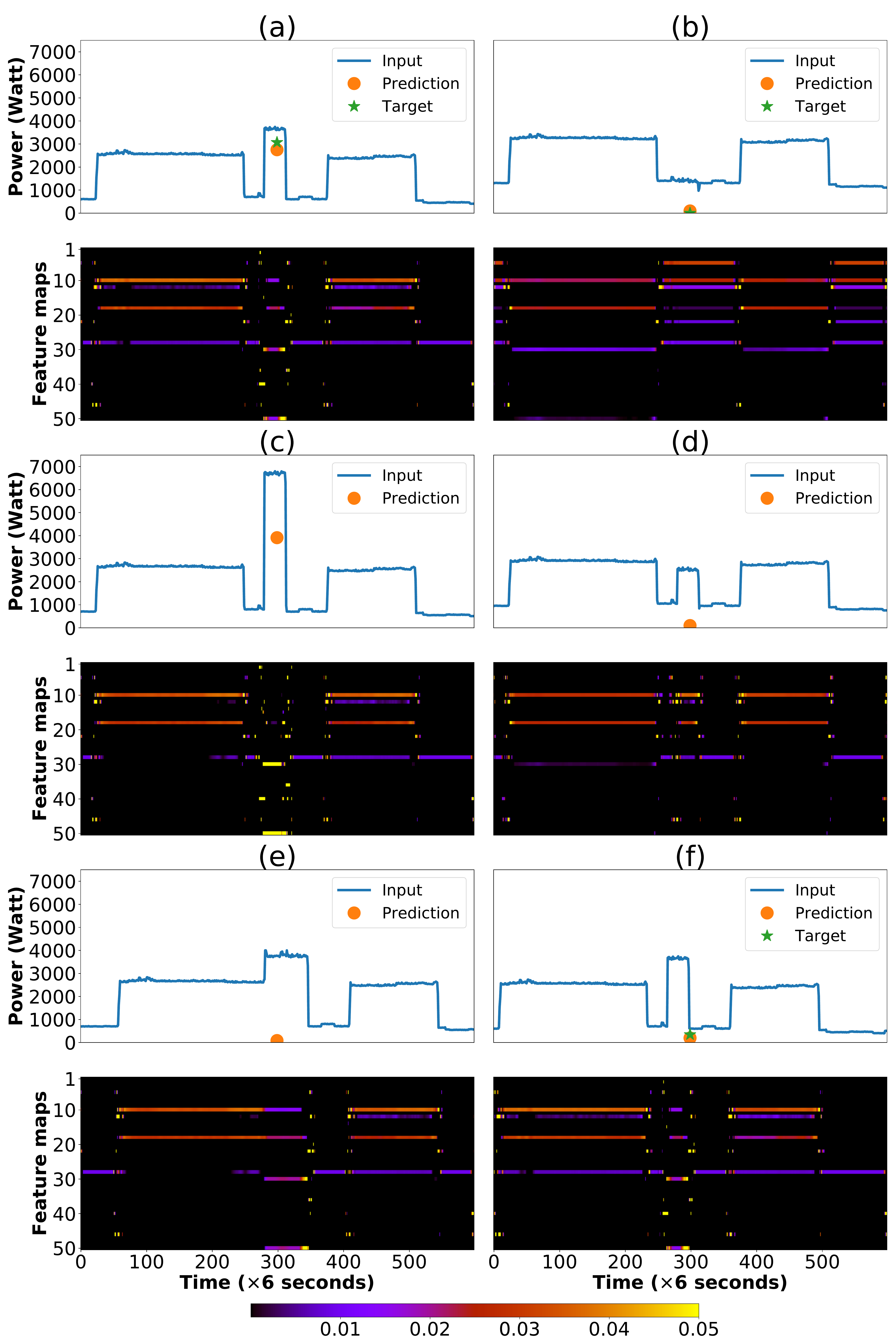}
\end{center}
\caption{\label{Fig:analysis} Feature maps learnt by the convolutional networks for various types of inputs (the mains). The feature maps contain the signatures of an appliance, which are used to predict the states of the appliance. These plots indicate that the learnt signatures are when an appliance is turned on and off (see the change points of the feature maps), the duration of an appliance when it is turned on, and the power level (yellow indicates higher power level). (a) The kettle is in the mains; the network detects the change points and power levels. (b) The kettle was manually removed from the mains;  comparing to (a), the change points and power levels were changed in the middle. (c) The power level of the kettle was set be double of the true level; comparing to (a), the detected power levels were increased in the middle. (d) The power level of the kettle was set to be half of the true level; comparing to (a), detected power levels were changed. (e) The duration of the kettle was set to be double; comparing to (a) the duration was changed. (f) Target has no activation at midpoint; the learnt signatures have the similar types to those in (a).}
\end{figure*}

\section{Conclusions}

We have proposed a sequence-to-point learning with neural networks for energy disaggregation. We have applied the proposed schemes to real world data sets. We have shown that sequence-to-point learning outperforms previous work using sequence-to-sequence learning. By visualizing the learnt feature maps, we have shown that the neural networks learn meaningful features from the data, which are crucial signatures for performing energy disaggregation. It would be interesting to apply the proposed methods to the single-channel blind source separation problems in other domains, for example, audio and speech.

\bibliographystyle{aaai}
\bibliography{seq2pointNilm}

%
%
%
%
%
%
%

\end{document}